\documentclass[12pt]{article}

\usepackage{geometry}
\usepackage{amsthm,amsmath,amssymb}
\usepackage{graphicx}
\usepackage{enumerate}

\usepackage[colorlinks=true,citecolor=black,linkcolor=black,urlcolor=blue]{hyperref}

\theoremstyle{plain}
\newtheorem{theorem}{Theorem}
\newtheorem{lemma}[theorem]{Lemma}
\newtheorem{corollary}[theorem]{Corollary}
\newtheorem{proposition}[theorem]{Proposition}

\theoremstyle{definition}
\newtheorem{definition}[theorem]{Definition}
\newtheorem{example}[theorem]{Example}

\theoremstyle{remark}
\newtheorem{remark}[theorem]{Remark}

\newtheorem{assumption}[theorem]{Assumption}

\usepackage{tikz}
\usetikzlibrary{positioning,chains,fit,shapes,calc}
\usepackage{times}
\usepackage{comment}
\usepackage{relsize}

\newcommand{\N}{\mathbb{N}}

\newcommand{\R}{\mathbb{R}}

\usepackage{ifthen}
% ...

\title{\bf Comparing the Switch and Curveball Markov Chains for Sampling Binary Matrices with Fixed Marginals}

\author{Corrie Jacobien Carstens \\ 
\small Korteweg-de Vries Institute for Mathematics\\[-0.8ex]
\small  University of Amsterdam\\[-0.8ex] 
\small Amsterdam, The Netherlands\\
\small\tt c.j.carstens@uva.nl\\
\and
Pieter Kleer \\ 
\small Centrum Wiskunde \& Informatica (CWI)\\[-0.8ex]
\small Amsterdam, The Netherlands\\
\small\tt kleer@cwi.nl
}

\begin{document}

\maketitle

\begin{abstract}
The Curveball algorithm is a variation on well-known switch-based Markov Chain Monte Carlo approaches for the uniform sampling of binary matrices with fixed row and column sums. 
We give a spectral gap comparison between switch chains and the Curveball chain using a decomposition of the switch chain based on Johnson graphs. In particular, this comparison allows us to prove that the Curveball Markov chain is rapidly mixing whenever one of the switch chains is rapidly mixing. 
As a by-product of our analysis, we show that the switch Markov chain of the Kannan-Tetali-Vempala conjecture  only has non-negative eigenvalues if the sampled binary matrices have at least three columns. This shows that the Markov chain does not have to be made lazy, which is of independent interest. 

\end{abstract}

\section{Introduction}
The problem of uniformly sampling binary matrices with fixed row and column sums (marginals) has received a lot of attention, see, e.g., \cite{Rao1996,Kannan1999,Erdos2013,Erdos2015,Erdos2016}. 
Equivalent formulations for this problem are the uniform sampling of undirected bipartite graphs, or the uniform sampling of directed graphs with possible a self-loop at every node (but no parallel edges).
One approach is to define a Markov chain on the space of all binary matrices for given fixed row and column sums, and study a random walk on this space induced by making small changes to a matrix using a given probabilisitic procedure (that defines the transition matrix). The idea, roughly speaking, is that after a sufficient amount of time, the so-called \emph{mixing time}, the resulting matrix almost corresponds to a sample from the uniform distribution over all binary matrices with  given row and column sums. The most well-known probabilistic procedures for making these small changes use so-called switches, see, e.g., \cite{Rao1996}. 
More recently the Curveball algorithm was introduced in some experimental papers
\cite{Verhelst2008,Strona2014}, which is a procedure that intuitively speeds up the mixing time of switch-based chains in many settings. 
The goal of this paper is to confirm this intuition by giving a spectral gap comparison for the Markov chains of the classical switch algorithm of Kannan, Tetali and Vempala \cite{Kannan1999} and the Curveball algorithm as formulated by Verhelst \cite{Verhelst2008}. 
We will start with an informal description of both algorithms.

For a given initial binary matrix $A$, in every step of the switch algorithm we randomly choose two distinct rows and two distinct columns uniformly at random. If the $2 \times 2$ submatrix corresponding to these rows and columns is a \emph{checkerboard} $C_i$ for $i = 1,2$, where, 
$$
C_1 = \left(\begin{matrix}
1 & 0 \\ 0 & 1
\end{matrix} \right) \ \ \ \ \text{ and }  \ \ \ \ 
C_2 = \left(\begin{matrix}
0 & 1 \\ 1 & 0
\end{matrix} \right),
$$
then the $2 \times 2$ submatrix is replaced by $C_{i+1}$ for $i$ modulo $2$. That is, if the checkerboard is $C_1$, it is replaced by $C_2$, and vice versa. If the submatrix does not correspond to a checkerboard, nothing is changed. Such an operation is called a \emph{switch}. 


The Curveball algorithm intuitively speeds up the switch algorithm. In every step of the algorithm, first two rows are chosen uniformly at random from $A$ as in the switch algorithm. Then, a so-called \emph{binomial trade} is performed. In such a trade, we first look at all the columns in the $2 \times n$ submatrix given by the chosen rows, and we identify all the columns for which the column sum, in this submatrix, is one. That is, the column consist of precisely one $1$ and one $0$. For example if the $2 \times 6$ submatrix (i.e., $n = 6$) is given by
$$
\begin{pmatrix}
1 & \mathbf{1} & 0 & \mathbf{0} & \mathbf{0} & \mathbf{1} \\ 1 & \mathbf{0} & 0 & \mathbf{1} & \mathbf{1} & \mathbf{0}
\end{pmatrix} ,
$$
then we consider the (auxiliary) submatrix 
$$
\left(\begin{matrix}
1 & 0 & 0 & 1 \\ 0 & 1 & 1 & 0
\end{matrix} \right)
$$
given by the second, fourth, fifth and sixth column. Let $u$ and $l$ respectively be the number of columns where the $1$ appears on the upper row and the lower row ($u = l = 2$ here)	. We now uniformly at random draw a $2 \times (u+l)$ matrix with columns sums equal to $1$, and row sums equal to $u$ and $l$. Note that there are $\binom{u+l}{u}$ possible choices, hence the name binomial trade. We then replace the (auxiliary) submatrix with this new submatrix in $A$. Note that such a drawing can be obtained by uniformly choosing $u$ out of $u+l$ column indices.



Both these algorithms define a Markov chain on the set of all $m \times n$ binary matrices satisfying given row and columns sums $r$ and $c$. The main result of this work is a comparison of their relaxation times, or, equivalently, spectral gaps (see next section for definitions).
\begin{theorem}[Relaxation time comparison]\label{thm:relax_comparison}
Let $(1 - \lambda_*^c)^{-1}$ and $(1 - \lambda_*^s)^{-1}$,  be the relaxation times of the Curveball and switch Markov chains respectively. Then, with $r_{\max}$ the maximum row sum,
$$
\frac{2}{n(n-1)}\cdot (1 - \lambda_*^s)^{-1} \ \leq \ (1 - \lambda_*^c)^{-1} \ \leq \ \min\left\{1, \frac{(2r_{\max} + 1)^2}{2n(n-1)}\right\} \cdot (1 - \lambda_*^s)^{-1}.
$$
\end{theorem}

\noindent We present a more general comparison framework inspired by, and based on, the notion of a heat-bath Markov chain as introduced by Dyer, Greenhill and Ullrich \cite{Dyer2014}. We prove Theorem \ref{thm:relax_comparison} as an application of this framework in the the more general setting where the binary matrices can also have \emph{forbidden entries} that must be zero. This allows us to also compare the chains for the sampling of a simple directed graph with given degree sequence, as its adjacency matrix can be modeled by a square binary matrix with zeros on the diagonal. 

\subsection{Related work}
Before going into related work, we would also like to refer the reader to \cite{Erdos2016} for a nice exposition on related work concerning the switch Markov chain.
Kannan, Tetali and Vempala \cite{Kannan1999} conjectured that the KTV-switch chain is rapidly mixing for all fixed row and column sums. 
Mikl\' os, Erd\H{o}s  and Soukup \cite{Erdos2013} proved the conjecture for half-regular binary matrices, in which all the row sums are equal (or all column sums), and Erd\H{o}s, Kiss, Mikl\' os and Soukup \cite{Erdos2015}  extended this result to almost half-regular marginals. The authors prove this in a slightly more general context where there might be certain forbidden edge sets.
The Curveball algorithm was first described by Verhelst \cite{Verhelst2008} and a slightly different version was later independently formulated by Strona, Nappo, Boccacci, Fattorini and San-Miguel-Ayanz \cite{Strona2014}. The name Curveball algorithm was introduced in \cite{Strona2014}. Theorem \ref{thm:relax_comparison} directly implies that the Curveball Markov chain is also rapidly mixing for (almost) half-regular marginals.

For the uniform sampling of simple directed graphs with a given degree sequence, the most used switch algorithm is the \emph{edge-switch} version,\footnote{We will address this version as well.} see Greenhill \cite{Greenhill2011}, who gives a polynomial upper bound on the mixing time for the case of $d$-regular directed graphs, and Greenhill and Sfragara \cite{Greenhill2017} for some recent results on certain irregular degree sequences. 
The latter paper \cite{Greenhill2017} only considers degree sequences for which the edge-switch Markov chain is irreducible for a given degree sequence. The Curveball chain has also been formulated for (un)directed graphs, see Carstens, Berger and Strona \cite{Carstens2016}. A theoretical analysis for the mixing time of the Curveball Markov chain was raised as an open problem there. 


All the results regarding rapid mixing mentioned above rely on the multi-commodity flow method developed by Sinclair \cite{Sinclair1992}. In this work we omit multi-commodity flow techniques in order to compare the switch and Curveball Markov chains, but rather take a more elementary approach based on comparing eigenvalues of transition matrices. 
One seeming advantage of the eigenvalue comparison is that it allows us to compare the switch and Curveball chains for arbitrary fixed row and column sums. 

Our spectral gap comparisons are special cases of the classical comparison framework developed largely by Diaconis and Saloff-Coste and is based on so-called Dirichlet form comparisons of Markov chains, see, e.g., \cite{Diaconis1993b,Diaconis1993a}, and also Quastel \cite{Quastel1992}.  See also the expository paper by Dyer, Goldberg, Jerrum and Martin \cite{Dyer2006}.  As the stationary distributions are the same for all our Markov chains, we use a more direct, but equivalent, framework based on positive semidefiniteness. We briefly elaborate on this in Appendix \ref{app:dirichlet} for the interested reader. 

The transition matrix of the Curveball Markov chain is a special case of a heat-bath Markov chain, as introduced by Dyer, Greenhill and Ullrich \cite{Dyer2014}. Our work partially builds on \cite{Dyer2014} in the sense that we compare a Markov chain, with a similar decomposition property as in the definition of a heat-bath chain, to its heat-bath variant. We explain these ideas in the next section. 

\section{General framework}\label{sec:general}
We consider an ergodic Markov chain $\mathcal{M} = (\Omega,P)$ with stationary distribution $\pi$, being strictly positive for all $x \in \Omega$,  that is of the form\footnote{This description is almost the same as that of a heat-bath chain \cite{Dyer2014}, and is introduced to illustrate the conceptual idea.}
\begin{equation}\label{eq:general_chain}
P = \sum_{a \in \mathcal{L}} \rho(a) \sum_{R \in \mathcal{R}_a} P_{R}
\end{equation}
which is given by a
\begin{enumerate}[i)]
\item finite index set $\mathcal{L}$,
and probability distribution $\rho$ over $\mathcal{L}$,
\item partition $\mathcal{R}_a = \cup R_{k,a}$ of $\Omega$ for $a \in \mathcal{L}$.
\end{enumerate}
Moreover, the restriction of a matrix $P_R$ to the rows and columns of $R = R_{k,a}$ defines the transition matrix of an ergodic, time-reversible Markov chain on $R$ (and is zero elsewhere), with stationary distribution 
$$
\tilde{\pi}_R(x) = \frac{\pi(x)}{\pi(R)}
$$ 
for $x \in R$. We use $1 = \lambda_0^R \geq  \lambda_1^R \geq \dots \geq \lambda_{|R|-1}^R$ to denotes its eigenvalues. Note that these are also eigenvalues of $P_R$ and that all other eigenvalues of $P_R$ are zeros (as all rows and columns not corresponding to elements in $R$ only contain zeros). We use $\mathcal{R}$ to denote the multi-set $\cup_a \mathcal{R}_a$ indexed by pairs $(k,a)$.  
Note that the chain $\mathcal{M}$ proceeds by drawing an index $a$ from the set $\mathcal{L}$, and then performs a transition in the Markov chain on the set $R$ that the current state is in.

 The \emph{heat-bath variant} $\mathcal{M}_{heat}$ of the chain $\mathcal{M}$ is given by the transition matrix 
\begin{equation}\label{eq:heat_chain}
P_{heat} = \sum_{a \in \mathcal{L}} \rho(a) \sum_{R \in \mathcal{R}_a} \mathbf{1}\cdot \sigma_R
\end{equation}
with $\sigma_R$ is a row-vector  given by $\sigma_R(x) = \tilde{\pi}_R(x)$ if $x \in R$ and zero otherwise, and $\mathbf{1}$ the all-ones column vector. 
It can be shown that $\mathcal{M}_{heat}$ is an ergodic Markov chain as well. It is reversible by construction \cite{Dyer2014}.\footnote{The Curveball chain is the heat-bath variant of the KTV-switch chain as we will later prove.}

\begin{theorem}\label{thm:heat_comparison}
Let $\mathcal{M}$ be a Markov chain as in (\ref{eq:general_chain}), and $\mathcal{M}_{heat}$ its heat-bath variant as in (\ref{eq:heat_chain}). If $\alpha$ and  $\beta$ are non-zero constants, with $\alpha\cdot \beta > 0$, such that
\begin{equation}
\label{eq:assumption}
\min_{R \in \mathcal{R}} \ \  \min_{i = 1,\dots,R-1} \{ \lambda_i^R, \alpha - \beta(1 - \lambda_i^R) \} \geq 0,
\end{equation}
then
\begin{equation}
\label{eq:alpha_beta}
\frac{1}{\alpha}\frac{1}{1 - \lambda_*^{heat}} \leq \frac{1}{\beta}\frac{1}{1 - \lambda_*},
\end{equation}
where $\lambda_*^{(heat)}$ is the second largest eigenvalue of $P_{(heat)}$. In particular, if $\lambda_{R-1}^R \geq 0$ for every $R \in \mathcal{R}$, then
$
(1 - \lambda_*^{heat})^{-1} \leq (1 - \lambda_*)^{-1}
$.
\end{theorem}

The intuition behind Theorem \ref{thm:heat_comparison} is that in order to compare the relaxation times of a Markov chain and its heat-bath variant, it suffices to compare them locally on the sets $R$. Note that $\alpha$ and $\beta$ can both be negative, so that this statement can  be used to lower bound the relaxation time of the heat-bath variant in terms of the original relaxation time as well.  
We will use the following propositions in the proof of Theorem \ref{thm:heat_comparison}. For $S \subseteq \Omega$, the matrix $I_S$ is defined by $I_S(x,x) = 1$ if $x \in S$ and zero otherwise. Also, a symmetric real-valued matrix $A$ is positive semidefinite if all its eigenvalues are non-negative, and this is denoted by $A \succeq 0$.

\begin{proposition}[\cite{Zhang1999}]\label{prop:eigenvalue_comparison}
Let $X,Y$ be symmetric $l \times l$ matrices. If $X - Y \succeq 0$, then $\lambda_i(X) \geq \lambda_i(Y)$ for $i = 1,\dots, l$, where $\lambda_i(C)$ is the $i$-th largest eigenvalue of $C = X,Y$.
\end{proposition}

\begin{proposition}
\label{prop:eigen_difference}
Let $X$ be the $k \times k$ transition matrix of an ergodic reversible Markov chain with stationary distribution $\pi$, and eigenvalues $1 = \lambda_0 > \lambda_1 \geq \dots \geq \lambda_{k-1}$. Let $X^* = \lim_{t \rightarrow \infty} X^t$ be the matrix containing the row vector $\pi$ on every row. Then the eigenvalues of $\alpha(I - X^*) - \beta(I -  X)$ are 
$$
\{0\} \cup \{ \alpha - \beta(1 - \lambda_i) \ \big| \  i = 1,\dots,k-1\} .
$$
for given constants $\alpha$ and $\beta$.
\end{proposition}
\begin{proof}
As $X$ is the transition matrix of a reversible Markov chain, it holds that the matrix $V X V^{-1}$ is symmetric, where $V = \text{diag}(\pi_1^{1/2},\pi_2^{1/2},\dots,\pi_k^{1/2}) = \text{diag}(\sqrt{\pi})$.\footnote{This is the same argument for showing that a reversible Markov chain only has real eigenvalues.} Note that the eigenvalues of $\alpha(I - X^*) - \beta(I -  X)$ are the same as those of 
$$
V(\alpha(I - X^*) - \beta(I -  X))V^{-1} = \alpha(I - \sqrt{\pi}^T\sqrt{\pi}) - \beta(I - VXV^{-1}).
$$
Moreover, with $1 = (1,1,1,\dots,1)^T$ the all-ones vector, we have
$$
VXV^{-1}\sqrt{\pi}^T = V X \mathbf{1} = V \mathbf{1} = \sqrt{\pi}^T,
$$
so that $\sqrt{\pi}^T$ is an eigenvector of $VXV^{-1}$ with eigenvalue $1$. It then follows that $\sqrt{\pi}^T$ is an eigenvector of $\alpha(I - \sqrt{\pi}^T\sqrt{\pi}) - \beta(I - VXV^{-1})$ with eigenvalue $0$. Let $\sqrt{\pi}^T = w_0, w_1,\dots,w_{k-1}$ be a basis of orthogonal eigenvectors for $VXV^{-1}$ corresponding to eigenvalues $\lambda_1,\dots,\lambda_{k-1}$ (note that $X$ and $VXV^{-1}$ have the same eigenvalues). It then follows that
$$
[\alpha(I - \sqrt{\pi}^T\sqrt{\pi}) - \beta(I - VXV^{-1})] w_i = \alpha - \beta(1 - \lambda_i)
$$
because of orthogonality. This completes the proof. 
\end{proof}

\begin{proof}[Proof of Theorem \ref{thm:heat_comparison}] 
Let $D$ be the $|\Omega| \times |\Omega|$ diagonal matrix with $(D)_{xx} = \sqrt{\pi(x)}$. As the matrices $\mathbf{1} \cdot \sigma_R$ and $P_R$ define reversible Markov chains on $R$, the matrix 
$$
Y_R = D^{-1}[\alpha(I_R - \mathbf{1}\cdot \sigma_R) - \beta(I_R - P_{R})]D
$$
is symmetric. Moreover, from the assumption in (\ref{eq:assumption}), together with Proposition \ref{prop:eigen_difference} and the fact that similar\footnote{Two square matrices $A$ and $B$ are \emph{similar} if there exists an invertible matrix $T$ such that $A = T^{-1}BT$.} matrices have the same set of eigenvalues, it follows that $Y_R$ is positive semidefinite. Since any non-negative linear combination of positive semidefinite matrices is again positive semidefinite, the matrix 
$$
D^{-1}[\alpha(I - P_{heat}) - \beta(I - P)]D = \sum_{a \in \mathcal{L}} \rho(a) \sum_{R \in \mathcal{R}_a} D^{-1}[\alpha(I_R - \mathbf{1}\cdot \sigma_R) - \beta(I_R - P_{R})]D
$$ 
is also positive semidefinite. Using Proposition \ref{prop:eigenvalue_comparison}, and again  the fact that similar matrices $A,B$ have the same set of eigenvalues, it follows that
$$
\alpha(1 - \lambda_i^{heat}) \geq \beta (1 - \lambda_i)
$$
where $\lambda_i^{(heat)}$ is the $i$-th largest eigenvalue of $P_{(heat)}$. Note that $P$ has non-negative eigenvalues as $D^{-1} P D$ is a non-negative linear combination of positive semidefinite matrices. A similar argument holds for $P_{heat}$ and was shown in \cite{Dyer2014}. In particular, it follows that $\lambda_1^{(heat)}$ is the second-largest eigenvalue of $P_{(heat)}$. This proves (\ref{eq:alpha_beta}).
\end{proof}

%

\subsection{Markov chain definitions} \label{sec:background}
Let $\mathcal{M} = (\Omega,P)$ be an ergodic, time-reversible Markov chain over state space $\Omega$ with transition matrix $P$ and stationary distribution $\pi$. We write $P_x^t = P^t(x,\cdot)$ for the distribution over $\Omega$ at time step $t$ given that the initial state is $x \in \Omega$. It is well-known that the matrix $P$ only has real eigenvalues $1 = \lambda_0 > \lambda_1 \geq \lambda_2 \geq \dots \geq \lambda_{N-1} > -1$, where $N = |\Omega|$. Moreover, we define $\lambda_{*} = \max\{\lambda_1,|\lambda_{N-1}|\}$ is the second-largest eigenvalue of $P$. The \emph{variation distance} at time $t$ with initial state $x$ is
$$
\Delta_x(t) = \max_{S \subseteq \Omega} \big| P^t(x,S) - \pi(s)\big| = \frac{1}{2}\sum_{y \in \Omega} \big| P^t(x,y) - \pi(y)\big|
$$
and the mixing time $\tau(\epsilon)$ is defined as
$$
\tau(\epsilon) = \max_{x \in \Omega}\left\{ \min\{ t : \Delta_x(t') \leq \epsilon \text{ for all } t' \geq t\}\right\}.
$$
A Markov chain is said to be \emph{rapidly mixing} if the mixing time can be upper bounded by a function polynomial in $\ln(|\Omega|/\epsilon)$. It is well-known, e.g., following directly from Proposition 1 \cite{Sinclair1992}, that
\begin{equation}\label{eq:mixing_time}
\frac{1}{2}\frac{\lambda_*}{1 - \lambda_*} \ln(1/2\epsilon) \ \leq \  \tau(\epsilon) \ 
\leq \ \frac{1}{1 - \lambda_*}\cdot (\ln(1/\pi_*) + \ln(1/\epsilon))
\end{equation}
where $\pi_* = \min_{x \in \Omega} \pi(x)$. This roughly implies that the mixing time is determined  by the \emph{spectral gap} $(1 - \lambda_*)$, or its inverse, the \emph{relaxation time} $(1 - \lambda_*)^{-1}$. 

We also introduce some additional notation. We let $G_{\Omega} = (\Omega,A)$ be the state space graph, with an arc $(a,b) \in A$ if and only if $P(a,b) > 0$ for $a,b \in \Omega$ with $a \neq b$. If $P$ is symmetric, we define $H_{\Omega} = (\Omega,E)$ as the undirected counterpart of $G_{\Omega}$ with $\{a,b\} \in E$ if and only if $(a,b),(b,a) \in A$ with $a \neq b$. Moreover, the $\delta$-lazy version of $\mathcal{M}$ is the Markov chain defined by transition matrix $(1 - \delta)I + \delta P$ for $0 < \delta < 1$. Note that this chain is also ergodic, and time-reversible with stationary distribution $\pi$.

\begin{proposition}\label{prop:lazy}
If $0 < \delta < 1$ is such the transition matrix $(1- \delta)I + \delta P$ of the $\delta$-lazy version of $\mathcal{M}$ only has non-negative eigenvalues. Then
$$
\frac{1}{1 - \lambda_{*,\delta}} \leq \frac{1}{\delta} \frac{1}{1 - \lambda_{*}}
$$
 where $\lambda_{*,\delta} = \lambda_{1,\delta} = (1 - \delta) + \delta \lambda_{1}$ is the second-largest eigenvalue of $(1- \delta) + \delta P$.
\end{proposition}
\begin{proof}
If $\lambda_i$ is an eigenvalue of $P$ then $\lambda_{i,\delta} := (1 - \delta) + \delta \lambda_i$ is an eigenvalue of $(1 - \delta)I + \delta P$. Note that $\lambda_i \leq \lambda_j$ if and only if $\delta \lambda_i \leq \delta \lambda_j$, which is true if and only if 
$$
\lambda_{i,\delta} = (1 - \delta) + \delta \lambda_i \leq (1 - \delta) + \delta \lambda_j = \lambda_{j,\delta}. 
$$
This in particular shows that $\lambda_{1,\delta} =  (1 - \delta) + \delta \lambda_{1}$ is indeed the second-largest eigenvalue of $(1 - \delta)I + \delta P$. Moreover, $\lambda_{i,\delta} = (1 - \delta) + \delta \lambda_i$ is equivalent to
$$
\frac{1}{1 - \lambda_{i,\delta}} = \frac{1}{\delta} \frac{1}{1 - \lambda_{i}}
$$
for $i > 0$. As the eigenvalues of $(1- \delta) + \delta P$ are all non-negative, we have
$$
\frac{1}{1 - \lambda_{*,\delta}} = \frac{1}{1 - \lambda_{1,\delta}} = \frac{1}{\delta} \frac{1}{1 - \lambda_{1}} \leq \frac{1}{\delta} \frac{1}{1 - \lambda_{*}}
$$
and this completes the proof. Note that the final inequality is true independent of the sign of $\lambda_1$.
\end{proof}


\subsection{Johnson graphs.}\label{sec:johnson}

One class of graphs that are of particular interest in this work, are the so-called Johnson graphs.  
For given integers $1 \leq q \leq p$, the undirected Johnson graph $J(p,q)$ contains as nodes all subsets of size $q$ of $\{1,\dots,p\}$, and two subsets $u,v \subseteq \{1,\dots,p\}$ are adjacent if and only if $|u \cap v| = q - 1$. We refer the reader to \cite{Holton1993,Brouwer2011} for the following facts. The Johnson graph $J(p,q)$ is a $q(p-q)$-regular graph and the eigenvalues of its adjacency matrix are given by
$$
(q - i)(p - q -i) - i \ \ \ \  \text{ with multiplicity }\ \ \ \   \binom{p}{i} - \binom{p}{i-1}
$$
for $i = 0,\dots,q$, with the convention that $\binom{p}{-1} = 0$. The following observation is included for ease of reference. It will often be used to lower bound the smallest eigenvalue of a Johnson graph.

\begin{proposition}\label{prop:johnson}
Let $p,q \in \N$ be given. The continuous function $f : \R \rightarrow \R$ defined by 
$$
f(x) = [(q - x)(p - q - x) - x] - q(p - q) = x(x - (p+1))
$$ 
is minimized for $x^* = (p+1)/2$ and $f(x^*) = -(p+1)^2/4$.
\end{proposition}

\section{Binary matrices and the switch chain}
We are given $n,m \in \N$, fixed row sums $r = (r_1,\dots, r_m)$, column sums $c = (c_1,\dots,c_n)$, and a set of forbidden entries $\mathcal{F} \subseteq \{1,\dots,m\} \times \{1,\dots,n\}$. The state space $\Omega = \Omega(r,c,\mathcal{F})$ is the set of all binary $m \times n$-matrices $A$ satisfying these row and column sums, and for which $A(a,b) = 0$ if $(a,b) \in \mathcal{F}$. For $A \in \Omega$, we let $A_{ij}$ be the $2 \times n$-submatrix formed by rows $i$ and $j$, for $1 \leq i < j \leq m$. We define 
\begin{equation}\label{eq:forbidden_T}
U_{ij}(A) = \{k \in \{1,\dots,n\} : A(i,k) = 1,\ A(j,k) = 0 \text{ and } (j,k) \notin \mathcal{F}\},
\end{equation}
with $u_{ij}(A) = |U_{ij}(A)|$, and similarly 
\begin{equation}\label{eq:forbidden_B}
L_{ij}(A) = \{k \in \{1,\dots,n\} : A(i,k) = 0,\ A(j,k) = 1 \text{ and } (i,k) \notin \mathcal{F}\},
\end{equation}
with $l_{ij}(A) = |L_{ij}(A)|$. Note that $L_{ij} \cup U_{ij}$ are precisely the columns $k$ for which $A_{ij}$ has different values on its rows and for which neither $(i,k)$ or $(j,k)$ is forbidden. 
Matrices $A, B \in \Omega$ are \emph{switch-adjacent for row $i$ and $j$} if $A = B$ or if $A - B$ contains exactly four non-zero elements that occur on rows $i$ and $j$, and the columns $k$ and $l$ containing these non-zero elements do not have forbidden entries in $A_{ij}$. 
Two matrices are switch-adjacent if they are switch-adjacent for some rows $i$ and $j$.

\emph{$\gamma$-Switch chain.} We next introduce the notion of a $\gamma$-switch Markov chain which is done for notational convenience as there are multiple switch-based chains available in the literature. For feasible $\gamma > 0$, the transition matrix of such a chain on state space $\Omega = \Omega(r,c,\mathcal{F})$ is given by
$$
P_{\gamma}(A,B) = \left\{ \begin{array}{ll} \binom{m}{2}^{-1}\cdot \gamma &  \ \ \ \ \ \ \text{if } A \neq B \text{ are switch-adjacent}, \\
\binom{m}{2}^{-1} \sum_{1 \leq i < j \leq m} 1 - u_{ij}l_{ij} \cdot \gamma & \ \ \ \ \ \ \text{if } A = B, \\
0 &\ \ \ \ \ \  \text{otherwise,}
\end{array}\right.
$$
provided $\gamma$ satisfies the following assumption.

\begin{assumption}\label{assump:gamma}
For given $n,m,r,c$ and $\mathcal{F}$, we assume that $\gamma$ is such that 
$$
1 - u_{ij}(A)l_{ij}(A)\cdot \gamma > 0
$$
for all $A \in \Omega$ and $1 \leq i < j \leq m$.
\end{assumption}
Note that the transition probability for switch-adjacent matrices is the same everywhere in the state space, and does not depend on the matrices. In particular, the transition matrix $P_{\gamma}$ is symmetric and hence the chain is reversible with respect to the uniform distribution. The factor $2/(m(m-1))$ is included for notational convenience. The chain can roughly be interpreted as follows. We first choose two distinct rows $i$ and $j$ uniformly at random, and then transition to a different matrix switch-adjacent for rows $i$ and $j$, of which there are $u_{ij}l_{ij}$ possibilities, where every matrix has probability $\gamma$ of being chosen;  and with probability $1 - u_{ij}l_{ij}\gamma$ we do nothing. Taking $\gamma = 2/(n(n-1))$ we get back the KTV-switch chain \cite{Kannan1999}. We will later show that  (a lazy version of) the edge-switch chain in \cite{Greenhill2011,Greenhill2017} also falls within this definition. \medskip

\begin{remark}
We always assume that the set $\Omega(r,c,\mathcal{F})$ is non-empty, and that the $\gamma$-switch chain is irreducible (it is clearly always aperiodic and finite).  
Irreducibility is in particular guaranteed in the case there are no forbidden entries \cite{Rao1996}; or in case $n = m \geq 4$, with $\mathcal{F}$ is the set of diagonal entries, and regular marginals $c_i = r_i = d$ for some given $d \geq 1$ \cite{Greenhill2011}. Note that the condition of irreducibility is independent of the value of $\gamma$.
\end{remark}\medskip

We next explain that the $\gamma$-switch chain is of the form (\ref{eq:general_chain}). The index set 
$$
\mathcal{L} = \{(i,j) : 1 \leq i < j \leq m\}
$$ 
is the set of all pairs of distinct rows, and $\rho$ is the uniform distribution over $\mathcal{L}$, that is, $\rho(a) =  \binom{m}{2}^{-1}$ for all $a \in \mathcal{L}$. The partitions $\mathcal{R}_a$ for $a \in \mathcal{L}$ rely on the notion of a binomial neighborhood, that is also defined in \cite{Verhelst2008} to describe the Curveball Markov chain (the decomposition idea  given here is novel). 

\begin{definition}[Binomial neighborhood]
For a fixed binary matrix $A$ and row-pair $(i,j)$, the $(i,j)$-binomial neighborhood $\mathcal{N}_{ij}(A)$ of $A$ is the set of matrices that can be reached by only applying switches on rows $i$ and $j$. More formally, $N_{ij}(A)$ contains all binary matrices $B \in \Omega$ for which $A(k,l) = B(k,l)$ whenever $(k,l) \notin \{i,j\} \times U_{ij}(A) \cup L_{ij}(A)$, which in particular implies that $U_{ij}(A) \cup L_{ij}(A) = U_{ij}(B) \cup L_{ij}(B)$.\footnote{Said differently, $\mathcal{N}_{ij}(A)$ contains all matrices that can be reached by one trade on rows $i$ and $j$ in the Curveball algorithm, as described in the introduction.}
\end{definition}

It should be clear that two matrices $A, B \in \Omega$ can be part of \emph{at most} one common binomial neighborhood, see also  \cite{Verhelst2008}. This follows directly from the observation that if $B \in \mathcal{N}_{ij}(A) \setminus \{A\}$, then $A$ and $B$ differ on precisely rows $i$ and $j$, so switches using any other pair of rows $\{k,l\} \neq \{i,j\}$ can never transform $A$ into $B$. Moreover, we have $A \in \mathcal{N}_{ij}(A)$;  if $B \in \mathcal{N}_{ij}(A)$, then $A \in \mathcal{N}_{ij}(B)$ \cite{Verhelst2008}; and, if $A \in \mathcal{N}_{ij}(B)$, $B \in \mathcal{N}_{ij}(C)$, then $A \in \mathcal{N}_{ij}(C)$. That is, the relation $\sim_{ij}$ defined by $a \sim b$ if and only if $a \in \mathcal{N}_{ij}(b)$, is an equivalence relation on $\Omega$. The equivalence classes of $\sim_{ij}$ define the set $\mathcal{R}_{(i,j)}$.
Finally, note that $u_{ij}(A) = u_{ij}(B)$ and $l_{ij}(A) = l_{ij}(B)$ if $A$ and $B$ are part of the same binomial neighborhood $\mathcal{N}$. Therefore, these numbers are only neighborhood-dependent, and not element-dependent within a fixed neighborhood. Observe that
$$
|\mathcal{N}| = \binom{u_{ij} + l_{ij}}{u_{ij}}.
$$

Moreover, another important observation is that the undirected state space graph (see Section \ref{sec:background}) $H$ of the $\gamma$-switch chain (which is the same for all $\gamma$) induced on a binomial neighborhood is isomorphic to a Johnson graph $J(u+l,u)$  whenever $u,l \geq 1$ (see Section \ref{sec:johnson} for notation and definition). If either $u = 0$ or $l = 0$ it consists of a single binary matrix. To see this, note that every element in the $(i,j)$-binomial neighborhood $\mathcal{N}_{ij}(A)$ can be represented by the set of indices of the columns $k$ for which $A(i,k) = 1, A(j,k) = 0$ and $(j,k) \notin \mathcal{F}$, which we denote by $Z(A_{ij})$. The set $\{1,\dots,l_{ij} + u_{ij}\}$ here is then the set indices of \emph{all} columns with precisely one $1$ and one $0$ on rows $i,j$ and that do not contain forbidden entries. Indeed, matrices $A \neq B$ are switch-adjacent for rows $i$ and $j$ if $Z(A_{ij}) \cap Z(B_{ij}) = u_{ij} - 1$.
Informally, the Markov chain resulting from always deterministically choosing rows $i$ and $j$ in the switch algorithm, is the disjoint union of smaller Markov chains each with a state space graph isomorphic to some Johnson graph.


\begin{example}\label{exmp:binom_neighborhood}
Consider the binary matrix
$$
A = \left(\begin{matrix}
0 & 1 & 1 & 0 & 1 & 0 & 1\\ 1 & 0 & 0 & 1 & 1 & 0 & 1 \\ 0 & 1 & 0 & 0 & 0 & 1 & 1
\end{matrix} \right)
$$ 
and the $2 \times 7$-submatrix formed by rows $1$ and $2$, which is
$$
A_{12} = \left(\begin{matrix}
0 & 1 & 1 & 0 & 1 & 0 & 1\\ 1 & 0 & 0 & 1 & 1 & 0 & 1 
\end{matrix} \right).
$$
For sake of simplicity, we (uniquely) describe every element of the $(1,2)$-binomial neighborhood $\mathcal{N}_{12}(A)$  by the first four columns (precisely those with column sums equal to one in the submatrix). For the switch chain, the induced subgraph of the undirected state space graph $H$ on the $(1,2)$-binomial neighborhood of $A$, the Johnson graph $J(4,2)$ is given in Figure \ref{fig:johnson}. 

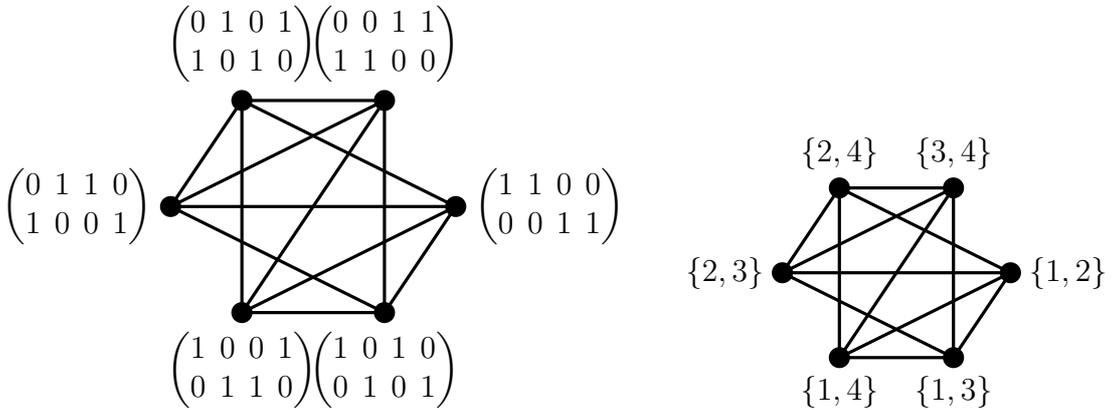
\begin{figure}[h!]
\centering
\begin{tikzpicture}[scale=3.75]
\coordinate (A1) at (-0.25,0); 
\coordinate (A2) at (0.25,0);
\coordinate (M1) at (-0.5,0.375);
\coordinate (M2) at (0.5,0.375);
\coordinate (T1) at (-0.25,0.75);
\coordinate (T2) at (0.25,0.75);


\node at (A1) [circle,scale=0.7,fill=black] {};
\node (a1) [below=0.1cm of A1]  {$\left(\begin{matrix}
 1 & 0 & 0 & 1 \\ 0 & 1 & 1 & 0
\end{matrix} \right)$};
\node at (A2) [circle,scale=0.7,fill=black] {};
\node (a2) [below=0.1cm of A2]  {$\left(\begin{matrix}
 1 & 0 & 1 & 0 \\ 0 & 1 & 0 & 1
\end{matrix} \right)$};
\node at (M1) [circle,scale=0.7,fill=black] {};
\node (m1) [left=0.1cm of M1]  {$\left(\begin{matrix}
 0 & 1 & 1 & 0 \\ 1 & 0 & 0 & 1
\end{matrix} \right)$};
\node at (M2) [circle,scale=0.7,fill=black] {};
\node (m2) [right=0.1cm of M2]  {$\left(\begin{matrix}
 1 & 1 & 0 & 0 \\ 0 & 0 & 1 & 1
\end{matrix} \right)$};
\node at (T1) [circle,scale=0.7,fill=black] {};
\node (t1) [above=0.1cm of T1]  {$\left(\begin{matrix}
 0 & 1 & 0 & 1 \\ 1 & 0 & 1 & 0
\end{matrix} \right)$};
\node at (T2) [circle,scale=0.7,fill=black] {};
\node (t2) [above=0.1cm of T2]  {$\left(\begin{matrix}
 0 & 0 & 1 & 1 \\ 1 & 1 & 0 & 0
\end{matrix} \right)$};

\path[every node/.style={sloped,anchor=south,auto=false}]
(T1) edge[-,very thick] node {} (T2) 
(T1) edge[-,very thick] node {} (M2) 
(T1) edge[-,very thick] node {} (A1) 
(T1) edge[-,very thick] node {} (M1) 
(T2) edge[-,very thick] node {} (M1) 
(T2) edge[-,very thick] node {} (A1) 
(T2) edge[-,very thick] node {} (A2) 
(M1) edge[-,very thick] node {} (A2)
(M1) edge[-,very thick] node {} (M2)
(A1) edge[-,very thick] node {} (M2)
(A1) edge[-,very thick] node {} (A2)
(A2) edge[-,very thick] node {} (M2);  
\end{tikzpicture}
\quad
\begin{tikzpicture}[scale=3]
\coordinate (A1) at (-0.25,0); 
\coordinate (A2) at (0.25,0);
\coordinate (M1) at (-0.5,0.375);
\coordinate (M2) at (0.5,0.375);
\coordinate (T1) at (-0.25,0.75);
\coordinate (T2) at (0.25,0.75);


\node at (A1) [circle,scale=0.7,fill=black] {};
\node (a1) [below=0.1cm of A1]  {$\{1,4\}$};
\node at (A2) [circle,scale=0.7,fill=black] {};
\node (a2) [below=0.1cm of A2]  {$\{1,3\}$};
\node at (M1) [circle,scale=0.7,fill=black] {};
\node (m1) [left=0.1cm of M1]  {$\{2,3\}$};
\node at (M2) [circle,scale=0.7,fill=black] {};
\node (m2) [right=0.1cm of M2]  {$\{1,2\}$};
\node at (T1) [circle,scale=0.7,fill=black] {};
\node (t1) [above=0.1cm of T1]  {$\{2,4\}$};
\node at (T2) [circle,scale=0.7,fill=black] {};
\node (t2) [above=0.1cm of T2]  {$\{3,4\}$};

\path[every node/.style={sloped,anchor=south,auto=false}]
(T1) edge[-,very thick] node {} (T2) 
(T1) edge[-,very thick] node {} (M2) 
(T1) edge[-,very thick] node {} (A1) 
(T1) edge[-,very thick] node {} (M1) 
(T2) edge[-,very thick] node {} (M1) 
(T2) edge[-,very thick] node {} (A1) 
(T2) edge[-,very thick] node {} (A2) 
(M1) edge[-,very thick] node {} (A2)
(M1) edge[-,very thick] node {} (M2)
(A1) edge[-,very thick] node {} (M2)
(A1) edge[-,very thick] node {} (A2)
(A2) edge[-,very thick] node {} (M2);  
\end{tikzpicture}
\caption{The induced subgraph $H$ for the switch chain on the $(1,2)$-binomial neighborhood of $A$. On the left we have indexed the nodes by the submatrices of the first four columns, and on the right by label sets, indicating the positions of the $1$'s on the top row (i.e., row $1$).}
\label{fig:johnson}
\end{figure} 
\end{example}

\begin{remark}
A fixed binomial neighborhood is reminiscient of the Bernoulli-Laplace Diffusion model, see, e.g., \cite{Diaconis1987,Donnelly1994} for an analysis of this model. Here, there are two bins with resp. $k$ and $n- k$ balls,  and in every transition two randomly chosen balls, one from each bin, are interchanged between the bins.  Indeed, the state space graph is then a Johnson graph \cite{Donnelly1994}. The transition probabilities are different, due to the non-zero holding probabilities in the switch algorithm, but the eigenvalues of this Markov chain are related to the eigenvalues of the switch Markov chain on a fixed binomial neighborhood, see also \cite{Diaconis1987,Donnelly1994}.
\end{remark}
\medskip

For a binomial neighborhood $\mathcal{N} = \mathcal{N}_{ij}(A)$ for given $i < j$ and $A \in \Omega$, the undirected graph $H_{\mathcal{N}} = (\Omega, E_{\mathcal{N}})$ is the graph where $E_{\mathcal{N}}$ forms the edge-set of the Johnson graph $J(u_{ij} + l_{ij},u_{ij})$ on $\mathcal{N} \subseteq \Omega$, and where all binary matrices $B \in \Omega \setminus \mathcal{N}$ are isolated nodes. We use $M(H_{\mathcal{N}})$ do denote its adjacency matrix. The discussion above leads to the following result summarizing that the $\gamma$-switch chain is of the form (\ref{eq:general_chain}), and that its heat-bath variant is precisely the Curveball Markov chain as in \cite{Verhelst2008} defined by transition matrix
$$
P_c(A,B) = \left\{ \begin{array}{ll} \binom{m}{2}^{-1}\cdot \binom{u_{ij} + l_{ij}}{u_{ij}}^{-1} &  \ \ \ \ \ \ \text{if } B \in \mathcal{N}_{ij}(A) \setminus \{A\}, \\
\binom{m}{2}^{-1}\sum_{1 \leq i < j \leq m}  \binom{u_{ij} + l_{ij}}{u_{ij}}^{-1}  & \ \ \ \ \ \ \text{if } A = B, \\
0 &\ \ \ \ \ \  \text{otherwise.}
\end{array}\right.
$$
Roughly speaking, the Curveball chain is precisely the chain sampling uniform within a fixed binomial neighborhood. For $S \subseteq \Omega$, the identity matrix $I_{S}$ on $S$ is defined by $I_{S}(x,x) = 1$ if $x \in S$ and zero elsewhere, and the all-ones matrix $J_{S}$ on $S$ is defined by $J_S(x,y) = 1$ if $x,y \in S$ and zero elsewhere.

\begin{theorem}\label{thm:switch_decomposition} The transition matrix $P_\gamma$ of the $\gamma$-switch chain is of the form (\ref{eq:general_chain}) namely
\begin{equation}\label{eq:switch_decomposition}
P_\gamma = \sum_{1 \leq i < j \leq m} \binom{m}{2}^{-1} \sum_{\mathcal{N} \in \mathcal{R}_{(i,j)}}  (1 - u_{ij}l_{ij}\cdot\gamma)\cdot I_{\mathcal{N}} + \gamma \cdot M(H_\mathcal{N}).
\end{equation}
The heat-bath variant of the $\gamma$-switch chain  is given by the Curveball chain, and can be written as
\begin{equation}\label{eq:curveball_decomposition}
P_c = \sum_{1 \leq i < j \leq m} \binom{m}{2}^{-1} \sum_{\mathcal{N} \in \mathcal{R}_{(i,j)}}  \binom{u_{ij} + l_{ij}}{u_{ij}}^{-1} J_{\mathcal{N}}.
\end{equation}
\end{theorem}
\begin{proof}
The decomposition in (\ref{eq:switch_decomposition}) follows from the discussion above, and Assumption \ref{assump:gamma} guarantees that the matrix 
$$
(1 - u_{ij}l_{ij}\cdot\gamma)\cdot I_{\mathcal{N}} + \gamma \cdot M(H_\mathcal{N})
$$
indeed defines the transition matrix of a Markov chain for every $\mathcal{N}$. Moreover, remember that the $\gamma$-switch chain has uniform stationary distribution $\pi$ over $\Omega$. Indeed, for a binomial neighborhood $\mathcal{N} = \mathcal{N}_{ij}(A)$ for given $i < j$ and $A \in \Omega$, the vector $\sigma_{\mathcal{N}}$ as used in (\ref{eq:heat_chain}) is then given by
$$
\sigma_{\mathcal{N}}(x) = \frac{\pi(x)}{\pi(\mathcal{N})} = \frac{1}{|\Omega|} \cdot \frac{|\Omega|}{|\mathcal{N}|} = \frac{1}{|\mathcal{N}|} = \binom{u_{ij} + l_{ij}}{u_{ij}}^{-1}
$$
if $x \in \mathcal{N}$, and zero otherwise. This implies that 
$$
\mathbf{1}\cdot \sigma_{\mathcal{N}} = \binom{u_{ij} + l_{ij}}{u_{ij}}^{-1} J_{\mathcal{N}}.
$$
as desired.
\end{proof}

This completes our description of the $\gamma$-switch chain as a Markov chain of the form (\ref{eq:general_chain}) with heat-bath variant the Curveball chain. We next study two explicit $\gamma$-switch chains.

\subsection{KTV-switch chain}
The switch chain of the Kannan-Tetali-Vempala conjecture, as described in the introduction, can be obtained by setting $\gamma = 2/(n(n-1))$. As the product $u_{ij}(A)l_{ij}(A)$ can be at most $n^2/4$ for any $A \in \Omega$ and $1 \leq i < j \leq m$, we see that $\gamma$ satisfies Assumption \ref{assump:gamma}.

\begin{theorem}\label{thm:curveball_KTV}
Let $P_c$ and $P_{KTV}$ be the transition matrices of resp. the Curveball and KTV-switch Markov chains with $n \geq 3$. 
Then
$$
\frac{2}{n(n-1)}\cdot (1 - \lambda_*^{KTV})^{-1} \ \leq \ (1 - \lambda_*^c)^{-1} \ \leq \ \min\left\{1, \frac{(2r_{\max} + 1)^2}{2n(n-1)}\right\} \cdot (1 - \lambda_*^{KTV})^{-1},
$$
where $\lambda_*^{(KTV,c)} = \lambda_1^{(KTV,c)} $ is the second largest eigenvalue of $P_{(KTV,c)}$. In particular, $P_{KTV}$ only has non-negative eigenvalues.
\end{theorem}
\begin{proof}
Let $\mathcal{N} = \mathcal{N}_{ij}(A)$ for given $i < j$ and $A \in \Omega$. We apply Theorem \ref{thm:heat_comparison} for various pairs $(\alpha,\beta)$.

\emph{Case 1: $\alpha = \beta = 1$.} From (\ref{eq:assumption}) it follows that it suffices to show that for any binomial neighborhood $\mathcal{N}$ the submatrix of
$$
Y_{\mathcal{N}} = \left[1 - u_{ij}l_{ij}\cdot\binom{n}{2}^{-1}\right] I_{\mathcal{N}} + \binom{n}{2}^{-1}M(H_{\mathcal{N}})
$$
formed by the rows and columns of $\mathcal{N}$ only has non-negative eigenvalues.
For any eigenvalue $\lambda$ of this submatrix, we  have
$$
\lambda = 1 + (\mu - u_{ij}l_{ij})\binom{n}{2}^{-1}
$$
where $\mu = \mu(\lambda)$ is an eigenvalue of the Johnson graph $J(u_{ij} +l_{ij}, u_{ij})$ on $\mathcal{N}$. In particular, using Proposition \ref{prop:johnson} with $p = u_{ij} + l_{ij}$ and $q = u_{ij}$, we get
$
(\mu - u_{ij}l_{ij}) \geq -\frac{1}{4}(u_{ij} + l_{ij} + 1)^2 \geq -\frac{1}{4}(n + 1)^2 
$
using that $0 \leq u_{ij} + l_{ij} \leq n$. Therefore, when $n \geq 5$, we have
$$
\lambda \geq 1 - \frac{1}{2}\frac{(n+1)^2}{n(n-1)} \geq 0.
$$
The cases $n = 3,4$ can be checked with some elementary arguments. This is left to the reader. Note that, in particular, this implies that $P_{KTV}$ only has non-negative eigenvalues when $n \geq  3$.

\emph{Case 2: $\alpha = 1$ and $\beta = (2n(n-1))/((2r_{\max} + 1)^2)$.} Using similar notation as in the previous case,  we show that 
$$
\lambda = 1 - \beta \left(1 - \left(1 + \left(\mu - u_{ij} \cdot l_{ij}\right)\binom{n}{2}^{-1}\right)\right) = 1 + \beta(\mu - u_{ij} \cdot l_{ij})\binom{n}{2}^{-1} \geq 0
$$
for any $\mu = \mu(\lambda)$ that is an eigenvalue of the Johnson graph $J(u_{ij} + l_{ij}, u_{ij})$. Again, using Proposition \ref{prop:johnson} in order to lower bound the quantity $(\mu -u_{ij} \cdot l_{ij})$, we find
$$
1 + \beta \cdot \left(\mu -  u_{ij} \cdot l_{ij}\right)\binom{n}{2}^{-1} \geq 1 - \frac{\beta}{4}(u_{ij} + l_{ij} +1)^2\binom{n}{2}^{-1} \geq 1 - \frac{\beta}{4}(2r_{\max}+1)^2\binom{n}{2}^{-1} \geq 0,
$$
using the fact that $0 \leq u_{ij} + l_{ij} \leq 2r_{\max}$ and the choice of $\beta$. 

\emph{Case 3: $\alpha = -1$ and $\beta = -2/(n(n-1))$.} We have to show that
%
$$
\lambda = \binom{n}{2} \left(1 - \left(1 + \left(\mu - u_{ij} \cdot l_{ij}\right)\binom{n}{2}^{-1} \right)\right) - 1 = u_{ij} \cdot l_{ij} - \mu - 1 \geq 0
$$
for all
$$
\mu = \mu(k) = (u - k)(l - k) - k
$$
where $k = 1,\dots,u$. Note that the eigenvalue $u_{ij} \cdot l_{ij}$ for the case $k = 0$ yields the largest eigenvalue $1 = \lambda_0^{\mathcal{N}}$ of $Y_{\mathcal{N}}$, and does not have to be considered here. The maximum over $k = 1,\dots,u$ is then attained for $k = 1$, and we have
$
u_{ij} \cdot l_{ij}- \mu - 1 \geq u_{ij} \cdot l_{ij} - ((u_{ij} - 1)(l_{ij} - 1) - 1) - 1 = u_{ij} + l_{ij} - 1 \geq 0,
$ since $u_{ij}, l_{ij} \geq 1$.
\end{proof}

\subsection{Edge-switch chain}
In every step of the edge-switch algorithm, two matrix-entries $(i,a)$ and $(j,b)$ from the set $\{(c,d) : A(c,d) = 1\}$ are chosen uniformly at random. We refer to it as the edge-switch algorithm, as for the interpretation of uniformly sampling directed graphs (where every node can have at most one self-loop), it corresponds to choosing two distinct edges uniformly at random. If the $2 \times 2$ submatrix corresponding to rows $i,j$ and columns $a,b$ forms a checkerboard, and if $(i,b)$ and $(j,b)$ are not forbidden entries, the checkerboard is adjusted (similar as for the KTV-switch algorithm as described in the introduction). Note that 
$$
P_{edge}(A,B) = \binom{\rho}{2}^{-1} 
$$
if $A$ and $B$ are switch-adjacent, where $\rho = \sum_{i} r_i$ is the total number of ones in every binary matrix in $\Omega$. Note that
$$
\gamma   =  \binom{m}{2} \binom{\rho}{2}^{-1}
$$
in this case.
The analysis in the main part of this section implies that we can write
\begin{equation}\label{eq:switch_reformulated}
P_{edge} = \sum_{1 \leq i < j \leq m} \binom{m}{2}^{-1} \sum_{\mathcal{N} \in \mathcal{R}_{(i,j)}}  \left[1 - u_{ij}l_{ij}\cdot\binom{m}{2}\binom{\rho}{2}^{-1}\right] I_{\mathcal{N}} + \binom{m}{2}\binom{\rho}{2}^{-1}M(H_{\mathcal{N}})
\end{equation}
where $M(H_{\mathcal{N}})$ is the adjacency matrix of a Johnson graph for every $\mathcal{N}$. However, the matrix
\begin{equation}\label{eq:switch_S}
S_{\mathcal{N}} = \left[1 - u_{ij}l_{ij}\cdot\binom{m}{2}\binom{\rho}{2}^{-1}\right] I_{\mathcal{N}} + \binom{m}{2}\binom{\rho}{2}^{-1}M(H_{\mathcal{N}})
\end{equation}
does not necessarily define the transition matrix of a Markov chain on $\mathcal{N}$, as the holding probabilities might be negative.\footnote{In versions (v1,v2) we wrongfully claim that these matrices are stochastic, from which we conclude that $(1 - \lambda_*^c)^{-1} \leq 2 (1  - \lambda_*^{edge})^{-1}$. 
We fix this claim in Theorem \ref{thm:comparison_edge} at the cost of a polynomial factor. We can therefore still conclude that the Curveball chain is rapidly mixing whenever the edge-switch chain is rapidly mixing. The results on regular instances, given later on, remain unchanged.}
We circumvent this problem by making the edge-switch chain $\delta$-lazy for $\delta$ sufficiently small. This procedure can be carried out for any $\gamma$ that does not satisfy Assumption \ref{assump:gamma}, provided $\gamma$ is polynomially bounded.

\begin{theorem}\label{thm:comparison_edge}
There exists a non-negative $\delta = \text{poly}(n,m,\rho)^{-1}$ 
such that
$$
\frac{1}{1 - \lambda_*^c} \leq  \frac{1}{\delta}\cdot \frac{1}{1 - \lambda_*^{edge}}
$$
where $\lambda_*^{c, (edge)}$ is the second largest eigenvalue of $P_{c, (edge)}$.
\end{theorem}
\begin{proof}
Note that
\begin{eqnarray}
(1 - \delta)I + \delta P_{edge} &=&  \sum_{i < j } \binom{m}{2}^{-1} \sum_{\mathcal{N} \in \mathcal{R}_{(i,j)}} (1 - \delta) + \delta \cdot S_{\mathcal{N}}  \nonumber \\
& = & \sum_{i < j } \binom{m}{2}^{-1} \sum_{\mathcal{N} \in \mathcal{R}_{(i,j)}}  \left[1 - \delta \cdot u_{ij}l_{ij}\binom{m}{2}\binom{\rho}{2}^{-1}\right] I_{\mathcal{N}} + \delta\cdot \binom{m}{2}\binom{\rho}{2}^{-1}M(H_{\mathcal{N}}) \nonumber
\end{eqnarray}
so by taking, e.g., 
$$
\delta = \frac{1}{2}\left[ \frac{n^2}{4}\binom{m}{2}\binom{\rho}{2}^{-1}\right]^{-1}
$$
we see that the matrices $(1 - \delta) + \delta \cdot S_{\mathcal{N}}$ are  stochastic matrices with non-negative eigenvalues, as all holding probabilities are at least $1/2$. Here we also use the fact that $u_{ij}(A)l_{ij}(A) \leq n^2/4$ for all $A \in \Omega$ and $1 \leq i < j \leq m$. We may conclude that $(1 - \lambda_{*}^c)^{-1} \leq (1 - \lambda_{*,\delta}^{edge})^{-1} \leq  (1 - \lambda_{*}^{edge})^{-1}/\delta$ where we use
Proposition \ref{prop:lazy} in the last inequality. 
\end{proof}\medskip

For certain instances we can do better than the $\delta$ in the proof of the previous theorem.

\begin{theorem}\label{thm:regular}
Let $\Omega = \Omega(n,d,\mathcal{F})$ be the set of square $n \times n$ binary matrices with row and column sums equal to $d \in \N$, so that $\rho = nd$, and forbidden entries $\mathcal{F}$. Then
$$
(1 - \lambda_*^{c})^{-1} \leq \left(\frac{2d+1}{2d}\right)^2 (1 - \lambda_*^{edge})^{-1}
$$
\end{theorem}
\begin{proof}With $S_{\mathcal{N}}$ as in (\ref{eq:switch_S}) we have that any eigenvalue $\lambda$ of $S_{\mathcal{N}}$ is of the form
$$
\lambda = 1 + (\mu - u_{ij}l_{ij})\binom{n}{2}\binom{nd}{2}^{-1}
$$ 
where $\mu = \mu(\lambda)$ is an eigenvalue of the Johnson graph $J(u_{ij} + l_{ij},u_{ij})$. Proposition \ref{prop:johnson} shows that
$
(\mu - u_{ij}l_{ij}) \geq -(u_{ij}+ l_{ij} + 1)^2/4 \geq  - (2d + 1)^2/4,
$
using $0 \leq u_{ij} + l_{ij} \leq 2d$ in the last inequality. It then follows that
$$
1 + (\mu - u_{ij}l_{ij})\binom{n}{2}\binom{nd}{2}^{-1} = 1 - \frac{1}{4}\frac{(2d+1)^2n(n-1)}{nd(nd-1)} = 1 - \frac{1}{4}\frac{4d^2(n-1)}{d(nd-1)} - \frac{1}{4}\frac{(4d+1)(n-1)}{d(nd-1)}.
$$
Note that $d(n-1) \leq nd - 1$ for all $n, d \geq 1$, from which it follows that
$$
1 + (\mu - u_{ij}l_{ij})\binom{n}{2}\binom{m}{2}^{-1} \geq - \frac{1}{4}\frac{(4d+1)(n-1)}{d(nd-1)} \geq -\frac{1}{d} - \frac{1}{4d^2}
$$
for all $n \in \N$. This implies that
$(1/d+1/(4d^2))I_{\mathcal{N}} + S_{\mathcal{N}}$ is positive semidefinite. Rescaling, and rewriting, gives that
$$
\left[1 - \left(\frac{2d}{2d+1}\right)^2\right]I_{\mathcal{N}} + \left(\frac{2d}{2d+1}\right)^2 S_{\mathcal{N}}
$$
is a symmetric stochastic transition matrix with only non-negative eigenvalues, i.e., we can take 
$$
\delta = \left(\frac{2d}{2d+1}\right)^2.
$$
\end{proof}

\begin{corollary}\label{cor:regular}
Let $\Omega = \Omega(n,d,\mathcal{F})$ be the set of square $n \times n$ binary matrices with row and column sums equal to $d \in \N$, so that $\rho = nd$, and forbidden entries $\mathcal{F}$, and let $\lambda_{|\Omega|-1}^{edge}$ be the smallest eigenvalue of $P_{edge}$. Then
$$
(1 + \lambda_{|\Omega|-1}^{edge})^{-1} \leq \frac{4d^2}{4d^2 - 4d - 1} \leq \frac{5}{2}
$$
if $d \geq 2$.
\end{corollary}
\begin{proof}
In the proof of Theorem \ref{thm:regular} it was shown that
$$
\left( \frac{1}{d} + \frac{1}{4d^2}\right) I + P_{edge} =  \sum_{1 \leq i < j \leq m} \binom{m}{2}^{-1} \sum_{\mathcal{N} \in \mathcal{R}_{(i,j)}}  \left( \frac{1}{d} + \frac{1}{4d^2}\right)I_{\mathcal{N}}  + S_{\mathcal{N}} \succeq 0,
$$
and hence
$
\lambda_{|\Omega|-1}^{edge} \geq -\left( \frac{1}{d} + \frac{1}{4d^2}\right).
$
Rewriting this gives the result. 
\end{proof}

With $\mathcal{F}$ the set of diagonal entries, this improves a bound of $(1 + \lambda_{|\Omega|-1}^{edge})^{-1} \leq n^2d^2/4$ of Greenhill \cite{Greenhill2011} for the edge-switch chain for the sampling of simple directed regular graphs.

\section{Parallelism in the Curveball chain}
As a binary matrix is only adjusted on two rows at the time in the Curveball algorithm, one might perform multiple binomial trades in parallel on distinct pairs of rows \cite{Carstens2016}. To be precise, in every step of the so-called \emph{$k$-Curveball algorithm}, we choose a set of $k \leq \lfloor m/2 \rfloor$ disjoint pairs of rows uniformly at random and perform a binomial trade on every pair (see introduction). For $k = \lfloor m/2 \rfloor$ this corresponds to the Global Curveball algorithm described in \cite{Carstens2016}.  We show that the induced $k$-Curveball chain is of the form (\ref{eq:general_chain}). The index set $\mathcal{L} = \mathcal{L}_k$ is the collection of all sets containing $k$ pairwise disjoint sets of two rows, i.e., 
$$
\left\{\{(1_a,1_b),(2_a,2_b),\dots,(k_a,k_b)\} \ : \ 1_a,1_b,\dots,k_a,k_b \in [m], \ |\{1_a,1_b,2_a,2_b,\dots,k_a,k_b\}| = 2k \right\},
$$
and $\rho$ is the uniform distribution over $\mathcal{L}$. For a fixed collection $\kappa \in \mathcal{L}_k$, we define the $\kappa$-neighborhood $\mathcal{N}_{\kappa}(A)$ of binary matrix $A \in \Omega$ as the set of binary matrices $B \in \Omega$ that can be obtained from $A$ by binomial trade-operations (see introduction) only involving the row-pairs in $\kappa$. Formally speaking, we have $B \in \mathcal{N}_{\kappa}(A)$ if and only if there exist binary matrices $A_l$ for $l = 0,\dots,k-1$, so that
$$
A_{l+1} \in \mathcal{N}_{(l+1)_a,(l+1)_b}(A_l)
$$
where $A = A_0$ and $B = A_{k}$. Note that the matrices $A_l$ might not all be pairwise distinct, as $A$ and $B$ could already coincide on certain pairs of rows in $\kappa$. Also note that $u_{i_ai_b}(A) = u_{i_ai_b}(B)$ and $l_{i_ai_b}(A) = l_{i_ai_b}(B)$ if $B \in \mathcal{N}_{\kappa}(A)$ for $i = 1,\dots,k$. It is not hard to see that such a neighborhood is isomorphic to a Cartesian product $W_1 \times W_2 \times \dots \times W_k$ of finite sets $W_1,\dots,W_k$ with 
$$
|W_i| = \binom{u_{i_ai_b} + l_{i_ai_b}}{u_{i_ai_b}}.\footnote{That is, the elements of $W_i$ describe a matrix on row-pair $(i_a,i_b)$.}
$$
Moreover, the relation $\sim_{\kappa}$ defined by $a \sim_{\kappa} b$ if and only if $b \in \mathcal{N}_{\kappa}(a)$ defines an equivalence relation, and its equivalence classes give the set $\mathcal{R}_{\kappa}$. We now consider the following artificial formulation of the original Curveball chain: we first select $k$ pairs of distinct rows uniformly at random, and then we choose one of those pairs uniformly at random and apply a binomial trade on that pair. It should be clear that this generates the same Markov chain as when we directly select a pair of distinct rows uniformly at random.
For $\mathcal{N}_{\kappa} \in \mathcal{R}_{\kappa}$ the matrix $P_{\mathcal{N}_{\kappa}}$ restricted to the rows and columns in $\mathcal{N}_{\kappa}$ is then the transition matrix of a Markov chain over $W_1 \times \dots \times W_k$, where in every step we choose an index $i \in [k]$ uniformly at random and make a transition in $W_i$ based on the (uniform) transition matrix
$$
Q_{i} = \binom{u_{i_ai_b} + l_{i_ai_b}}{u_{i_ai_b}} ^{-1} J
$$
where $J$ is the all-ones matrix of approriate size. More formally, the  matrix $P_{\mathcal{N}_{\kappa}}$ restricted to the columns and rows in $\mathcal{N}_{\kappa}$ is given by
\begin{equation}\label{eq:trans_product}
\frac{\sum_{i = 1}^k \left[ \mathbf{\otimes}_{j = 1}^{i-1} \mathcal{I}_j\right] \otimes Q_{i}  \otimes \left[\otimes_{j = i+1}^k \mathcal{I}_j\right]}{k},
\end{equation}
forming a transition matrix on $\mathcal{N}_{\kappa}$, and is zero elsewhere. Here $\mathcal{I}_j$ is the identity matrix with the same size as $Q_j$ and $\otimes$ the usual tensor product. The eigenvalues of the matrix in (\ref{eq:trans_product}) are given by
\begin{equation}\label{eq:eigen_product}
\lambda_{\mathcal{N}_{\kappa}} =  \left\{\frac{1}{k} \sum_{i = 1}^k \lambda_{{j_i},i} : 0 \leq j_i \leq |W_i| - 1\right\}
\end{equation}
where $1= \lambda_{0,i} \geq \lambda_{1,i} \geq \dots \geq \lambda_{|W_i| - 1,i}$ are the eigenvalues of $Q_i$ for $i = 1,\dots,k$.\footnote{See, e.g., \cite{Erdos2015decomposition} for a similar argument regarding the transition matrix, and eigenvalues, of a Markov chain of this form. These statements follow directly from elementary arguments involving tensor products.} It then follows that
$$
P_{c} = \sum_{\kappa \in \mathcal{L}_k} \frac{1}{|\mathcal{L}_k|} \sum_{\mathcal{N}_{\kappa} \in \mathcal{R}_{\kappa}}  P_{\mathcal{N}_{\kappa}}
$$ 
which is of the form (\ref{eq:general_chain}). For $k = 1$, we get back the description of the previous section. Now, its heat-bath variant is precisely the $k$-Curveball Markov chain
$$
P_{k-Curveball} = \sum_{\kappa \in \mathcal{L}_k} \frac{1}{|\mathcal{L}_k|} \sum_{\mathcal{N}_{\kappa} \in \mathcal{R}_{\kappa}}  \frac{1}{|\mathcal{N}_{\kappa}|} J_{\mathcal{N}_{\kappa}},
$$
where
$$
|\mathcal{N}_{\kappa}| = \prod_{i = 1}^k \binom{u_{i_ai_b} + l_{i_ai_b}}{u_{i_ai_b}}^{-1}
$$
as, roughly speaking, for a fixed neighborhood $\mathcal{N}_{\kappa}$, the $k$-Curveball chain is precisely the uniform sampler over such a neighborhood.

\begin{theorem}\label{thm:global_curveball}
We have
$$
\frac{(1 - \lambda_*^{c})^{-1}}{k} \ \leq \ (1 - \lambda_*^{k,c})^{-1} \ \leq \  (1 - \lambda_*^{c})^{-1}
$$
where $\lambda_*^{k,c}$ is the second-largest eigenvalue of the $k$-Curveball chain, and $\lambda_*^c$ the second-largest eigenvalue of the $1$-Curveball chain.
\end{theorem}
\begin{proof}
The upper bound follows from Theorem \ref{thm:heat_comparison}, with $\alpha = \beta = 1$, as the eigenvalues of all the $Q_i$ are non-negative, and therefore (\ref{eq:eigen_product}) implies that the eigenvalues of the matrix in (\ref{eq:trans_product}) are also non-negative. For the lower bound, we take $\alpha = -1$ and $\beta = -k$. That is, we have to show that
$$
-1 + k(1 - \mu)
$$
with $\mu \in  \lambda_{\mathcal{N}_{\kappa}} \setminus \{1\}$ as in (\ref{eq:eigen_product}). It is not hard to see that the second-largest eigenvalue in $\lambda_{\mathcal{N}_{\kappa}}$ is $(k-1)/k$, as the eigenvalues of every fixed $Q_i$ are $1 = \lambda_{0,i} > \lambda_{1,i} = \dots = \lambda_{|W_i| - 1} = 0$. This implies that
$$
-1 + k(1 - \mu) \geq -1 + k(1 - (k-1)/k) \geq 0 
$$
for all $\mu \in  \lambda_{\mathcal{N}_{\kappa}} \setminus \{1\}$.
\end{proof}
In general, the upper bound is tight for certain (degenerate) cases, that is, parallelism in the Curveball chain does not necessarily guarantee an improvement in its relaxation time. E.g., take column marginals $c_i = 1$ for $i = 1,\dots,n$, and row-marginals $r_1 = r_2 = n/2$ and $r_3 = r_4 = 0$, and consider $k = 2$. 

\section{Conclusion}
We believe similar ideas as in this work can be used to prove that the Curveball chain is rapidly mixing  for the sampling of undirected graphs with given degree sequences \cite{Carstens2016}, whenever one of the switch chains is rapidly mixing for those marginals. We leave this for future work, as the proof we have in mind is a bit more involved, but of a very similar nature as the ideas described here. 
An interesting direction for future work is to give a better comparison between the edge-switch chain and Curveball chain. It would also be interesting to see if there exist classes of marginals for which one can give a strict improvement over the result in Theorem \ref{thm:global_curveball}.

%
%
%

\subsection*{Acknowledgements}
Pieter Kleer is grateful to Annabell Berger and Catherine Greenhill for some useful discussions and comments regarding this work.

 \bibliographystyle{plain}
\bibliography{references}

\appendix

\section{Markov chain comparison using Dirichlet forms}\label{app:dirichlet}
In this appendix we include some notes on the comparison framework for Markov chains based on Dirichlet forms and show that, for our setting, it is equivalent to a comparison in terms of positive semidefiniteness. The description is taken from Chapter 13.3 \cite{Levin2009}.

Let $\mathcal{M}$ be an ergodic, reversible Markov chain on state space $\Omega$ with transition matrix $P$ and stationary distribution $\pi$. The Dirichlet form for the pair $(P,\pi)$ is defined by
$$
\mathcal{E}(f,h) := \langle(I - P)f,h\rangle_{\pi} 
$$
for functions $f,h \in \{g \ \big| \ g : \Omega \rightarrow \R\}$, where $\langle g_1,g_2\rangle_{\pi} = \sum_{x \in \Omega} g_1(x)g_2(x)\pi(x)$. To illustrate the usefulness of Dirichlet forms, consider the following result, which appears, e.g., as Lemma 13.22 in \cite{Levin2009}. 

\begin{lemma}
Let $P$ and $\tilde{P}$ be reversible transition matrices with stationary distributions $\pi$ and $\tilde{\pi}$, respectively. If $\tilde{\mathcal{E}}(f,f) \leq \alpha \mathcal{E}(f,f)$ for all $f \in \{g \ \big| \ g : \Omega \rightarrow \R\}$, then
$$
1 - \tilde{\lambda}_1 \leq \left[\max_{x \in \Omega} \frac{\pi(x)}{\tilde{\pi}(x)} \right]\alpha (1 - \lambda_1),
$$
where $\lambda_1$ and $\tilde{\lambda}_1$ are resp. the second largest eigenvalue of $P$ and $\tilde{P}$. In particular, if both stationary distributions are the same, we get $1 - \tilde{\lambda}_1 \leq \alpha (1 - \lambda_1)$. 
\end{lemma}

The following proposition relates the Dirichlet form to the use of positive semidefinite matrices, in case both stationary distributions are the uniform distribution over $\Omega$. We can then essentially use the above lemma instead of Proposition \ref{prop:eigenvalue_comparison}. We choose to give Proposition \ref{prop:eigenvalue_comparison} as this avoids having to introduce the Dirichlet framework.

\begin{proposition}
Suppose that $\pi$ and $\tilde{\pi}$ are both the uniform distribution over $\Omega$. Then $\tilde{\mathcal{E}}(f,f) \leq \alpha \mathcal{E}(f,f)$ is equivalent to
$$
\alpha(I - P) \succeq (I - \tilde{P}).
$$
\end{proposition}
\begin{proof}

If both stationary distributions are the uniform distribution over $\Omega$, then the condition
\begin{equation}\label{eq:dirichlet}
\tilde{\mathcal{E}}(f,f) \leq \alpha \mathcal{E}(f,f)
\end{equation}
is equivalent to 
$$
f^T(I - \tilde{P})f \leq \alpha f^T(I - P)f
$$
where the function $f$ is interpreted as a vector. This in turn is equivalent to stating that $\alpha(I - P) \succeq (I - \tilde{P})$. This follows from the equivalence that $A \succeq 0$ if and only if $x^TAx \geq 0$ for all real-valued vectors $x$. 
\end{proof} 

\end{document}